\documentclass[10pt, reqno]{amsart}

\usepackage{amssymb, mathrsfs, paralist, booktabs, a4wide}
\usepackage{hyperref}

\title[Volume comparison and singularity theorems]
{Volume comparison for hypersurfaces in Lorentzian manifolds and singularity theorems}

\numberwithin{equation}{section}

\theoremstyle{plain}
\newtheorem{theorem}{Theorem}[section]
\newtheorem{corollary}[theorem]{Corollary}
\newtheorem{proposition}[theorem]{Proposition}
\newtheorem{lemma}[theorem]{Lemma}

\theoremstyle{definition}
\newtheorem{definition}[theorem]{Definition}

\theoremstyle{remark}
\newtheorem{remark}[theorem]{Remark}
\newtheorem{remarks}[theorem]{Remarks}

\newcommand\bel[1]{\begin{equation}\label{#1}}
\newcommand\ee{\end{equation}}

\newcommand{\gammad}{\ensuremath{\dot{\gamma}}}
\newcommand{\g}{\ensuremath{\mathbf{g}}}
\newcommand{\n}{\ensuremath{\mathbf{n}}}
\newcommand{\h}{\ensuremath{\mathbf{h}}}
\newcommand{\Cinfty}{\ensuremath{C^{\infty}}}
\newcommand{\R}{\ensuremath{\mathbb{R}}}
\newcommand{\D}{\operatorname{\mathrm{d}}\!}
\renewcommand\H{\ensuremath{H}}

\newcommand\vs{\vskip .2cm}

\newcommand{\At}[1] {\Big|_{#1}}
\newcommand{\id} {\operatorname{\mathsf{id}}}
\newcommand{\image} {\operatorname{{\mathrm{im}}}}
\newcommand{\SP}[1] {\left\langle{#1}\right\rangle}
\DeclareMathSymbol\dAlembert {\mathop}{AMSa}{"03}
\newcommand{\Sec}[1][{k}] {\Gamma^{#1}}
\newcommand{\Secinfty} {\Sec[\infty]}
\newcommand{\End} {\operatorname{\mathsf{End}}}
\newcommand{\norm}[1] {\left|{#1}\right|}
\newcommand{\Ric}{\ensuremath{\mathbf{Ric}}}
\newcommand{\Rm}{\ensuremath{\mathbf{R}}}
\newcommand{\E}{\mathrm{e}}
\newcommand{\CCC}{\ensuremath{\mathrm{CCC}}}

\newcommand{\Ical}{\ensuremath{\mathcal{I}}}
\newcommand{\Jcal}{\ensuremath{\mathcal{J}}}
\newcommand{\Scal}{\ensuremath{\mathcal{S}}}

\DeclareMathOperator{\Cut}{Cut}
\DeclareMathOperator{\supp}{supp}
\DeclareMathOperator{\gradient}{grad}
\DeclareMathOperator{\vol}{vol}
\DeclareMathOperator{\area}{area}
\DeclareMathOperator{\Hess}{Hess}
\DeclareMathOperator{\tr}{tr}

\begin{document}
\author{Jan-Hendrik\ Treude}
\address{Fakult{\"a}t f{\"u}r Mathematik \\ Universit{\"a}t Regensburg \\ D-93040 Regensburg\\ Germany}
\email{\href{mailto:jan-hendrik.treude@mathematik.uni-regensburg.de}{jan-hendrik.treude@mathematik.uni-regensburg.de}}
\urladdr{\href{http://homepages.uni-regensburg.de/~trj19484}{http://homepages.uni-regensburg.de/\lower3pt\hbox{\symbol{'176}}trj19484}}
\author{James~D.E.\ Grant}
\address{Gravitationsphysik \\ Fakult{\"a}t f{\"u}r Physik \\ Universit{\"a}t Wien \\ Boltzmanngasse 5 \\ 1090 Wien \\ Austria}
\email{\href{mailto:james.grant@univie.ac.at}{james.grant@univie.ac.at}}
\urladdr{\href{http://jdegrant.wordpress.com}{http://jdegrant.wordpress.com}}
\subjclass[2010]{53C23, 53C80.} \keywords{Lorentzian geometry,
comparison theorems, singularity theorems}
\date{January 19, 2012. Preprint UWThPh-2012-2} \thanks{The authors are grateful to Prof.~R.\
  Steinbauer for helpful conversations. JHT would like to thank Stefan
  Suhr and Olaf M{\"u}ller for clarifying discussions. The work of JDEG was initially supported by
  START-project Y237--N13 of the~\href{http://www.fwf.ac.at/}{Austrian Science Fund}. A preliminary
  version of this paper was prepared while JDEG was visiting
  the~\href{http://www.esi.ac.at}{Erwin Schr\"{o}dinger Institute} as part of the programme \lq\lq
  Dynamics of General Relativity\rq\rq.}
\begin{abstract}
We develop area and volume comparison theorems for the
evolution of spacelike, acausal, causally complete
hypersurfaces in Lorentzian manifolds, where one has a lower
bound on the Ricci tensor along timelike curves, and an upper
bound on the mean curvature of the hypersurface. Using these
results, we give a new proof of Hawking's singularity theorem.
\end{abstract}
\maketitle
\thispagestyle{empty}

%========================================================================================================

\section{Introduction}

There are many similarities between the ideas inherent in the
proof of the singularity theorems in Lorentzian geometry, and
those underlying the proofs of certain Riemannian comparison
theorems. For example, the interplay between Riccati techniques
and index techniques in both fields has been emphasised by
Ehrlich~\cite{EhrlichSlides}. The combination of completeness
(which guarantees minimizing geodesics) and curvature
conditions (which, via Riccati techniques, imply that geodesics
have conjugate points) are the key ingredients in the proof of,
for example, Myers's theorem in Riemannian geometry. It is,
similarly, the interplay between global hyperbolicity and
conjugate points which leads to the singularity theorems in
Lorentzian geometry.

We have two objectives in this paper. Firstly, we derive
comparison results concerning the area and volume of sets that
evolve from a fixed spacelike hypersurface in a Lorentzian
manifold. In particular, we prove area and volume monotonicity
theorems concerning such quantities, where compared with fixed
Lorentzian warped product manifolds. Our techniques are based
on ideas from Riemannian geometry (see, e.g.,~\cite{HeintzeKarcher}). Our second aim is to apply these
results to give a new proof of Hawking's singularity
theorem~\cite[pp.~272]{HE}. The idea of the proof is that
geometrical conditions required for Hawking's result (i.e. that
the Ricci tensor be non-negative on timelike vector fields and
the initial surface have negative mean curvature) are
sufficient to ensure that the volume of the future evolution of
the spacelike hypersurface is finite. Combining this property
with causal structure arguments then gives the result. Our
philosophy here is somewhat similar to that recently employed
in the metric measure spaces where Myers's theorem is deduced
from a generalised version of the Brunn--Minkowski
inequality~\cite{Sturm2, Lott.Villani:2009a}. In particular,
our approach was motivated by the wish to find a method of
proof of the singularity theorems that may be generalized to
the low-regularity Lorentzian setting.

\smallskip

The plan of the paper is as follows. After recalling necessary
background material in Section~\ref{sec:background}, we develop the
comparison results for Riccati equations that we require in
Section~\ref{sec:Riccati}. In Section~\ref{sec:RicciComparison}, we
apply these results to derive area and volume monotonicity results
for spacelike hypersurfaces in Lorentzian manifolds that satisfy what
we call the \emph{cosmological comparison condition}, $\CCC(\kappa,
\beta)$. (See Definition~\ref{def:cosmological-curvature-condition}.)
In particular, we introduce comparison geometries in which our area
and volume conditions are sharp, and show that geometries satisfying
the $\CCC(\kappa, \beta)$ condition satisfy monotonicity properties
relative to these model geometries. In
Section~\ref{sec:singularity-theorems}, we show how these geometrical
comparison theorems may be applied to give new proofs of the Hawking
singularity theorem~\cite[pp.~272]{HE}. Since one of our proofs is
based upon geometrical comparison arguments for volumes, areas, etc,
it seems plausible that it can be adapted to the low-regularity
regime, and, in particular, may be the basis for a proof of the
singularity theorems for metrics that are not $C^2$.%
\footnote{This possibility is investigated in~\cite{SynthSing}.}
After some final remarks, and an outline of some possible directions
for further research, in Appendix~\ref{sec:bddabove}, we investigate
some conditions under which we may prove a \emph{lower\/} bound on
the time separation between $\Sigma$ and focal points along normal
geodesics. This is essentially an adaption to the Lorentzian regime
of the Rauch comparison theorem for submanifolds in Riemannian
geometry given by Warner~\cite{Warner:Rauch}. It leads to comparison
theorems where areas and volumes of sets are bounded below in terms
of those of corresponding sets in a model geometry.

\subsection*{Notation}
Throughout, $(M, \g)$ will denote a connected $(n+1)$-dimensional
Lorentzian manifold. We will generally use the standard notation
of~\cite{BEE, HE, ONeill}. In particular, we adopt the convention
that the metric is of signature $(-, +, \ldots, +)$. We will also
denote the product $\g(u, v)$ by $\SP{u, v}$. The curvature tensor of
the metric is defined with the convention $\Rm(X, Y) Z = \left(
    [\nabla_X, \nabla_Y] - \nabla_{[X, Y]} \right) Z$, and we denote
the Ricci tensor of $\g$ by $\Ric$.

\section{Background material}
\label{sec:background}

\subsection{Causality theory}
We first review the concepts that we require from the causal
structure theory of Lorentzian manifolds. For extensive, modern
reviews of this material, see~\cite{piotr, Minguzzi.Sanchez}.

Let $(M, \g)$ be a connected $(n+1)$-dimensional Lorentzian
manifold. Let $p \in M$. A non-zero tangent vector, $v \in T_p
M$, is said to be \emph{timelike}, \emph{null}, or
\emph{spacelike\/} if $\SP{v, v} < 0$, $\SP{v, v} = 0$, or
$\SP{v, v} > 0$, respectively. A vector that is either timelike
or null is called \emph{causal}. These notions naturally extend
to vector fields.

For each $p \in M$, the set of causal vectors in $T_p M$ has
two connected components, the two \emph{causal cones}. A
\emph{time-orientation for $T_p M$\/} is the specification of
one of the two causal cones as the \emph{future causal cone},
and the other one as the \emph{past causal cone}. A
\emph{time-orientation for $M$\/} is a continuous choice of
time-orientation in all tangent spaces. A Lorentzian manifold
is either time-orientable, or it admits a double-cover that is
time-orientable. Therefore, without any important loss of
generality, we will assume throughout that our Lorentzian
manifolds are time-oriented. We will refer to a Lorentzian
manifold with time-orientation as a~\emph{spacetime}.

A piecewise smooth curve%
\footnote{By a curve, we will mean a continuous map $\gamma
\colon I \to M$, where $I \subseteq \R$ is an interval.} in $M$
is called~\emph{future-directed timelike\/} if its tangent
vector is timelike and lies in the future causal cone at all
points. Analogously, we define piecewise smooth past-directed
timelike curves and piecewise smooth future- and past-directed
causal curves. Given $p, q \in M$, we write $p \ll q$ if there
exists a piecewise smooth future-directed timelike curve from
$p$ to $q$. Similarly, we write $p < q$ if there exists a
piecewise smooth future-directed causal curve from $p$ to $q$.
Finally, we write $p \leq q$ if either $p < q$ or $p = q$. Let
$A \subseteq M$ be an arbitrary subset of $M$. We define the
\emph{chronological and causal future of $A$} to be the sets
\begin{align*}
  I^+(A) &:= \{ q \in M \mid \exists\, p \in A: p \ll q \} \, ,
  \\
  J^+(A) &:= \{ q \in M \mid \exists\, p \in A: p \leq q \} \, ,
\end{align*}
respectively. Analogously, we define the chronological and causal
past $I^-(A)$ and $J^-(A)$ of $A$.

Finally, we demand that $M$ be \emph{globally hyperbolic}, i.e. we
impose that the following two conditions hold (cf., e.g., ~\cite{BernalSanchez}):
\begin{compactenum}
\item $M$ is causal, i.e. we have $p \not< p$ for all $p \in M$.
\item For all $p, q \in M$, the causal diamond $J(p, q) := J^+(p) \cap J^-(q) \subset M$ is compact.
\end{compactenum}

\subsection{Time-separation and maximizing curves}

Let $\gamma \colon [a, b] \to M$ be a piecewise smooth curve. The
\emph{Lorentzian arc-length\/} of $\gamma$ is defined to be
\[
L(\gamma) = \sum_{i=1}^k \int_{t_{i-1}}^{t_i} \norm{\dot{\gamma}(t)} d t \,,
\]
where $a = t_0 < \ldots < t_k = b$ are the breakpoints of $\gamma$
(i.e. points at which the tangent vector of $\gamma$ is not
continuous), and $\norm{\dot{\gamma}(t)} := \sqrt{|\SP{\dot{\gamma}(t), \dot{\gamma}(t)}|}$.

\begin{definition}
 The \emph{time-separation\/} $\tau \colon M \times M \to [0, \infty]$ is defined by
 \bel{eq:time-separation}
  \tau(p, q) :=
  \sup \left\{ L(\gamma) \left|
    \begin{aligned}
     &\gamma \mbox{ is a piecewise
      smooth future- }
     \\
     &\mbox{directed causal curve from $p$ to
      $q$}
    \end{aligned}
   \right.
  \right\}
 \ee
 if $p < q$, and by $\tau(p, q) = 0$ if $p \not< q$. If the supremum
 in~\eqref{eq:time-separation} is attained by a piecewise smooth
 future-directed causal curve, $\gamma$, from $p$ to $q$, then
 $\gamma$ is said to be \emph{maximizing\/} between $p$ and $q$.
\end{definition}

The time-separation may be considered as a Lorentzian analogue of the
distance function in Riemannian geometry, although its properties
differ in several important respects (see, e.g.,~\cite[Chap.~4]{BEE}).
An important global question concerns the
existence of maximizing curves. For globally hyperbolic spacetimes, we
have the following well-known result, which illustrates why global
hyperbolicity of a Lorentzian manifold may be compared to completeness
of a Riemannian manifold from the point of view of arc-length.
\begin{theorem}
 \label{maxgeodspts}
 Let $(M, \g)$ be globally hyperbolic.
 \begin{compactenum}
 \item For all $p, q \in M$ with $p < q$, there exists a maximizing
  curve $\gamma \colon [a, b] \to M$ from $p$ to $q$. If $p \ll q$,
  then $\gamma$ may be reparametrized to be a timelike
  geodesic. Otherwise, $\gamma$ can be reparametrized to be a null
  geodesic. In both cases, the corresponding geodesic has no
  conjugate points prior to $q$.
\item The time-separation of $M$ is finite-valued and continuous.
 \end{compactenum}
\end{theorem}

We shall need a slight variant of the time-separation. Recall that a
subset $A \subseteq M$ is said to be \emph{acausal\/} if $p \not< q$
for all $p, q \in A$. In particular, if $A$ is acausal then $I^+(A)
\cap I^-(A) = \emptyset$, so we may introduce the \emph{signed
  time-separation to $A$,\/} $\tau_A \colon M \to \overline{\R}$, by
\[
\tau_A(q)
:=
\begin{cases}
    \sup_{p \in A} \tau(p, q) & q \in I^+(A) \\
    - \sup_{p \in A} \tau(q, p) & q \in I^-(A) \\
    0 & \mbox{else}
\end{cases} \,.
\]

In order for these suprema to be attained, global hyperbolicity alone
is not sufficient. Rather, one must demand additional compactness
properties for $A$. The following concept, introduced by
Galloway~\cite{Galloway}, is well-suited to this purpose.

\begin{definition}
    A subset $A \subseteq M$ is~\emph{future causally complete (FCC)},
    if for each $q \in J^+(A)$ the intersection $J^-(q) \cap A
    \subseteq A$ has compact closure in $A$. Similarly, one defines
    \emph{past causal completeness (PCC)}. A subset that is both FCC
    and PCC is called~\emph{causally complete}.
\end{definition}

\begin{remark}
    Clearly every compact set is causally complete. More
    interestingly, every acausal Cauchy hypersurface is causally
    complete (cf.~\cite[Lemma~14.40]{ONeill}). Conversely, one can
    show that every causally complete, acausal (topological)
    hypersurface in a globally hyperbolic spacetime is actually a
    Cauchy hypersurface. This fails if causal completeness is weakened
    to either only FCC or PCC, and counterexamples are provided by
    (spacelike) hyperboloids in Minkowski spacetime. Finally, we
    mention that in globally hyperbolic spacetimes, the notion of
    causal completeness coincides with the notion of \emph{causal
      compactness\/} as defined in~\cite[Def.~5.1.1]{Friedlander}.
    Without global hyperbolicity, causal completeness is a weaker
    condition than causal compactness.
\end{remark}

\begin{theorem}
  Let $(M, \g)$ be globally hyperbolic and $\Sigma \subset M$ a
  smooth, spacelike, acausal, FCC hypersurface. Then:
  \begin{compactenum}
  \item For each $q \in J^+(\Sigma)$, there exists a point $p \in
   \Sigma$ with $\tau_\Sigma(q) = \tau(p, q)$. Furthermore, the
   maximizing geodesic from $p$ to $q$ is timelike, normal to
   $\Sigma$ and has no focal points before $q$.
\item The signed time-separation of $\Sigma$ is finite-valued and
    continuous on $J^+(\Sigma)$.
  \end{compactenum}
  An analogous result holds for $\Sigma$ being PCC and $q \in
  J^-(\Sigma)$.
\end{theorem}
\begin{proof}
    Fix $q \in J^+(\Sigma)$. By global hyperbolicity, the function
    $\tau(\cdot, q) \colon M \to \R$ is continuous. Therefore, it
    attains its maximum on the compact subset\footnote{We denote by
      $\overline{A}$ the closure and by $A^\circ$ the interior of a
      subset $A \subset M$.} $K := \overline{J^-(q) \cap \Sigma}
    \subset \Sigma$ at some point $p \in K$. Since $\supp \tau(\cdot,
    q) \subset J^-(q)$, this implies that
 \[
 \tau(p, q)
 = \sup_{p' \in K} \tau(p', q)
 = \sup_{p' \in \Sigma} \tau(p', q) = \tau_{\Sigma}(q).
 \]
 The required properties of the maximizing curve are a standard
 result from the analysis of the index form (see,
 e.g.,~\cite[Chap.~10]{ONeill}).
\end{proof}

\subsection{Causal cut locus}
\label{subsec:causal-cut-locus}

Let $M$ be globally hyperbolic and $\Sigma \subset M$ a smooth,
spacelike, acausal, FCC hypersurface.
Let $N \Sigma \to \Sigma$ by the normal bundle of $\Sigma \subset M$,
and $\exp_\Sigma \colon N\Sigma \to M$ the normal exponential map.
We introduce the \emph{future unit-normal bundle}
\[
S^+N\Sigma
:= \{ v \in N\Sigma \mid \mbox{$v$ future-directed}, \SP{v, v} = -1 \} \,.
\]
For $v \in S^+N\Sigma$, denote by $\gamma_v \colon I_v \to M$ the
unique maximal geodesic with $\dot{\gamma}_v(0) = v$.
It can be shown that (see, e.g.,~\cite[Cor.~3.2.23]{treude:2011a}),
for each $v \in S^+N\Sigma$,
$\gamma_v$ maximizes the time-separation to $\Sigma$ for small parameter values,
in the sense that
\[
\tau_\Sigma(\gamma_v(t)) = L(\gamma_v|_{[0, t]})
\quad \mbox{ for all sufficiently small $t > 0$}.
\]
Therefore, for each $v \in S^+N\Sigma$, we have
\[
s^+_\Sigma(v) := \sup \{ t \in I_v \mid \tau_\Sigma(\gamma_v(t))
= L(\gamma_v|_{[0, t]}) \}
> 0\,.
\]
This defines a function
$s^+_\Sigma \colon S^+ N\Sigma \to (0, \infty]$,
called the \emph{$\Sigma$-future cut function}.
If $s_\Sigma^+(v) \in I_v$, then the point $\gamma_v(s_\Sigma^+(v))$
is called the \emph{$\Sigma$-cut point\/} of $\gamma_v$.
The collection of such points, i.e. the set
\[
\Cut^+(\Sigma)
:= \{ \exp_\Sigma(s_\Sigma^+(v)v) \mid v \in S^+N\Sigma \mbox{ and }
s_\Sigma^+(v) \in I_v \}
\subset M \,,
\]
is called the \emph{future cut locus\/} of $\Sigma$.

One can show that a point $q \in M$ lies in $\Cut^+(\Sigma)$ if and
only if either $q$ is a focal point of $\Sigma$, or if $q$ can be
connected to $\Sigma$ by more than one maximizing geodesic.
Furthermore, points of the second type are dense in the cut locus.%
\footnote{For proofs of these and the following statements,
  see~\cite[Sec.~3.2.5]{treude:2011a}.} Using these properties, one
has the following result.

\begin{theorem}
  \label{theorem:largest-normal-neighborhood}
  Let $\Sigma \subset M$ be a smooth, spacelike, acausal, FCC
  hypersurface. Let
  \[
  \Jcal^+_T(\Sigma) :=
  \{ tv \mid v \in S^+N\Sigma\;
  \mathrm{and}\; t \in [0, s^+_\Sigma(v)) \}
  \subset N\Sigma
  \]
  and $\Ical^+_T(\Sigma) = \Jcal^+_T(\Sigma)^\circ$.
  Then the following properties hold:
  \begin{enumerate}
  \item $\Ical^+(\Sigma) : = \exp_\Sigma(\Ical^+_T(\Sigma)) \subset M$ is open and diffeomorphic to $\Ical^+_T(\Sigma)$ via $\exp_\Sigma$.
  \item $\Ical^+(\Sigma) = I^+(\Sigma) \setminus \Cut^+(\Sigma)$.
  \item $\Cut^+(\Sigma) \subset M$ has measure zero and is closed.
  \item $\Ical^+(\Sigma)$ is the largest open subset of $I^+(\Sigma)$
      with the property that each of its points can be connected to
      $\Sigma$ by a unique maximizing geodesic.
  \end{enumerate}
\end{theorem}

Clearly, if $\Sigma \subset M$ is PCC instead of FCC,
then analogous properties hold with future sets replaced with past sets.

\subsection{Regularity of the time-separation}
\label{subsec:regularity-of-time-separation}

One can use Theorem~\ref{theorem:largest-normal-neighborhood}
to show that on $\Ical^+(\Sigma)$ the signed time-separation
$\tau_\Sigma$ is actually smooth rather than just continuous.
\begin{proposition}
 \label{proposition:signed-time-separation}
 Let $\Sigma \subset M$ be a smooth, spacelike, acausal, FCC hypersurface.
 Then the signed time-separation $\tau_\Sigma \colon M \to \R$
 is smooth on $\Ical^+(\Sigma)$ and has the following properties:
  \begin{enumerate}
  \item For each $q \in \Ical^+(\Sigma)$, we have%
   \footnote{Given a $C^1$ function, $f$, on $M$ the gradient
   of $f$ is the vector field defined by the relation
   $\SP{X, \gradient f} = X(f)$, for all vector fields $X$.}
   $\gradient \tau_\Sigma|_q = - \dot{\gamma}(\tau_\Sigma(q))$,
   where $\gamma \colon [0, \tau_\Sigma(q)] \to M$ is the unique
   maximizing geodesic from $\Sigma$ to $q$, parametrized to
   unit-speed.
  \item On $\Ical^+(\Sigma)$, the vector field $\gradient
   \tau_\Sigma$ is past-directed timelike and has
   unit-length. Furthermore, $\gradient \tau_\Sigma$ extends to a
   smooth unit normal for $\Sigma$.
  \end{enumerate}
\end{proposition}
\begin{proof}
    The first part is shown in a similar way to the analogous
    statement in Riemannian geometry (cf., e.g.,~\cite[Prop.~4.8]{sakai:1996a}). The second part is an immediate
    consequence of the first.
\end{proof}

\subsection{Level Sets of Distance Functions}
\label{sec:distance-function}

Adopting the standard terminology from Riemannian geometry, we will
refer to a smooth function $\tau \in \Cinfty(M)$ that satisfies
$\SP{\gradient \tau, \gradient \tau} = -1$ as a \emph{timelike distance function}.
Without loss of generality, we assume that $\gradient \tau$
is past-directed (otherwise, consider $-\tau$). The principal example
that will be of interest to us is the signed time-separation function
of a smooth, spacelike, acausal, FCC hypersurface $\Sigma \subset M$,
restricted to $\Ical^+(\Sigma)$ (cf.~Proposition~\ref{proposition:signed-time-separation}).

Given such a function $\tau$, a short calculation shows that
\bel{distfnparallel} \nabla_{\gradient \tau} \gradient \tau = 0 \,.
\ee As a consequence, integral curves of the vector field $\gradient
\tau$ are (past-directed, timelike, unit-speed) geodesics. In
addition, $\norm{\gradient \tau} = 1$ implies that the map $\tau
\colon M \to \R$ is a (semi-Riemannian) submersion. Thus the level
sets of this map are embedded, spacelike hypersurfaces, which we
denote by $\Scal_t := \tau^{-1}(\{t\}) \subset M$. The restriction of
$\gradient \tau$ to $\Scal_t$ is a past-directed unit-normal to
$\Scal_t$, so the vector field $\n := - \gradient\tau$ yields the
corresponding future-directed unit normal.

Consider the subbundle $T\Scal := \bigcup_{t \in \image(\tau)}
T\Scal_t \subset TM$, and let $\tan \colon TM \to T\Scal$ be the
corresponding orthogonal projection. For each $t \in \image(\tau)$,
let $S_t \in \Secinfty(\End(T\Scal_t))$ be the shape operator of
$\Scal_t \subset M$ with respect to the future-directed unit normal
$\n$, which we define with sign convention
\[
S_t(w)
:= \tan \left( \nabla_w \n \right)
= \nabla_w \n + \SP{\nabla_w \n, \n} \n = \nabla_w \n \,,
\qquad w \in T\Scal_t \,.
\]
Therefore, if we define $S \in \Secinfty(\End(TM))$ by $S(X) =
\nabla_X \n$ for $X \in \Secinfty(TM)$, then, for each $t \in
\image(\tau)$, the restriction of $S$ to $T\Scal_t$ is the shape
operator of $\Scal_t \subset M$ with respect to $\n$. In particular,
the corresponding (future) \emph{mean curvature\/} $\H_t := \tr S_t
\in \Cinfty(\Scal_t)$ is given by
\bel{eq:mean-curvature-of-level-set-of-distance-function} \H_t(q) =
\sum_{i=1}^n \SP{\nabla_{e_i} \n, e_i} = - \sum_{i=1}^n
\SP{\nabla_{e_i} \gradient \tau, e_i} = - \tr \Hess \tau|_q = -
\dAlembert \tau(q) \,, \ee where $q \in \Scal_t$, and $e_1, \ldots,
e_n \in T_q\Scal_t$ is an arbitrarily chosen orthonormal basis. The
third equality follows from~\eqref{distfnparallel} and the fact that
$\gradient \tau|_q, e_1, \ldots, e_n$ is an orthonormal basis of $T_q
M$.

It will be important to us to know how the family of shape-operators
$\{S_t\}_t$ changes with respect to the parameter $t$. More
precisely, it will be crucial that they obey the following Riccati
equation.

\begin{theorem}
  \label{theorem:shape-operator-riccati-equation}
  Let $S \in \Secinfty(\End(TM))$ be given by $S(X) = \nabla_X
  \n$ for $X \in \Secinfty(TM)$. Then
  \bel{Riccati}
    \nabla_{\n} S + S^2 + R_{\n} = 0 \, ,
  \ee
  where $S^2 = S \circ S$ is to be understood pointwise, and
  $R_{\n} \in \Secinfty(\End(TM))$ denotes the map $X \mapsto \Rm(X, \n) \n$.
\end{theorem}

\begin{proof}
    Let $X \in \Secinfty(TM)$, then we have
  \begin{align*}
    (\nabla_{\n} S)(X)
    &= \nabla_{\n} (S(X)) - S(\nabla_{\n} X)
    \\
    &= \nabla_{\n} \left( \nabla_X \n \right) - S(\nabla_{\n} X)
    \\
    &= \nabla_X \left( \nabla_{\n} \n \right)
    + \Rm(\n, X) \n
    + \nabla_{[\n, X]} \n - S(\nabla_{\n} X)
    \\
    &= - \Rm(X, \n) \n
    + S([\n, X] ) - S(\nabla_{\n} X)
    \\
    &= - R_{\n}(X) - S^2(X) \,.
  \end{align*}
  In the fourth equality, we have used
  equation~\eqref{distfnparallel} and, in the last step, the
  identity $[\n, X] = \nabla_{\n} X - \nabla_X \n$.
\end{proof}

Properties of solutions of equation~\eqref{Riccati} will be studied
in the next section. We conclude this section with two additional
results about distance functions that we will require.

\vs Let $\Phi \colon \mathcal{U} \subset \R \times M \to M$ be the
flow of $\n$, i.e. $\frac{\D}{\D t} \Phi_t(p) = \n(\Phi_t(p))$ for
$(t, p) \in \mathcal{U}$. For $p \in M$, we have
\[
\frac{\D}{\D t} \tau(\Phi_t(p))
= \D_{\Phi_t(p)}\tau( \n_{\Phi_t(p)})
= \SP{\n|_{\Phi_t(p)}, \gradient \tau|_{\Phi_t(p)}}
= 1 \,.
\]
This implies that for $K \subset S_t$ and $s \in \R$ such that $\{s\}
\times K \subset \mathcal{U}$, we have $\Phi_s(K) \subset S_{t+s}$.
Using this observation, one can show the following standard result:

\begin{proposition}[First Variation of Area]
  \label{proposition:first-variation-of-area}
  For $t \in \image(\tau)$, let $K \subset \Scal_t$ be
  compact and assume that the flow, $\Phi$, of $\n$ is defined on
  $[-\epsilon, \epsilon] \times K$ for some $\epsilon > 0$. Set
  $K_s := \Phi_s(K) \subset \Scal_{t+s}$ for each $s \in
  [-\epsilon, \epsilon]$. Then
  \bel{eq:first-variation-of-area}
    \frac{\D}{\D s}\At{s=0} \area K_s
    = \int_{K} \tr S_t \D \mu_t \,.
  \ee
  Here $\mu_t$ denotes the Riemannian volume measure of
  $(\Scal_t, g|_{\Scal_t})$.
\end{proposition}

Finally, we recall the following version of Fubini's theorem.

\begin{proposition}[Coarea Formula]
 \label{proposition:coarea-formula}
 For $f \in \mathcal{L}^1(M, \D\mu_g)$, we have $f|_{\Scal_t}
 \in \mathcal{L}^1(\Scal_t, \D \mu_t)$ for almost all $t \in
 \image(\tau)$ and
 \bel{eq:coarea-formula}
  \int_M f \D \mu_g
  = \int_{\R} \bigg( \int_{\Scal_t} f|_{\Scal_t}
  \D \mu_t \bigg) \D t
  \,.
 \ee
\end{proposition}

\section{Riccati comparison}
\label{sec:Riccati}

We now abstractly study some properties of solutions of the Riccati
equation~\eqref{Riccati}. Let $E$ be an $n$-dimensional, real vector
space with positive-definite inner product $\SP{\cdot, \cdot}$. Denote
by $\mathsf{S}(E) \subset \End(E)$ the subspace of linear maps $E \to E$ that
are self-adjoint with respect to $\SP{\cdot, \cdot}$. For $A, B \in
\mathsf{S}(E)$, we write $A \geq B$ if $A-B \geq 0$ in the sense that
$\SP{(A-B)v, v} \geq 0$ for all $v \in E$.

\

We will require the following result from~\cite{eschenburg.heintze:1990a}:

\begin{theorem}
  \label{theorem:riccati-comparison}
  Let $R_1, R_2 \colon \R \to \mathsf{S}(E)$ be smooth with $R_1 \ge R_2$, in the sense that $R_1(t)
  \geq R_2(t)$ for all $t \in \R$. Assume that for $i=1, 2$ we have a solution $S_i
  \colon (0, t_i) \to \mathsf{S}(E)$ of $S_i' + S_i^2 + R_i = 0$,
  which cannot be extended beyond $t_i$. If $U := S_2-S_1$ has a
  continuous extension to $t=0$ with $U(0) \geq 0$, then the
  following hold.
  \begin{enumerate}
  \item We have $t_1 \leq t_2$ and $S_1 \leq S_2$ on $(0, t_1)$.
  \item The function $d(t) := \dim \ker U(t)$ is monotonically
    decreasing on $(0, t_1)$.
  \item If $S_1(s) = S_2(s)$ for some $s \in (0, t_1)$, then on $(0,
    s]$ we have $S_1 = S_2$ and $R_1 = R_2$.
  \end{enumerate}
\end{theorem}

Let $R \colon \R \to \mathrm{S}(E)$ be smooth, and $S
\colon I' \to \mathsf{S}(E)$ a solution of the Riccati equation
\bel{eq:riccati-equation}
S' + S^2 + R = 0,
\ee
for some interval $I' \subseteq \R$. Using Theorem~\ref{theorem:riccati-comparison}, we now
show that a lower bound on $\tr R$ implies an upper bound on $\tr S$.

\begin{definition}
  \label{definition:expansion-vorticity-shear}
  Let $S \colon I' \to \End(E)$. We define the \emph{expansion\/}
  $\theta \in \Cinfty(I')$, the \emph{vorticity\/} $\omega \colon I'
  \to \End(E)$, and the \emph{shear\/} $\sigma \colon I' \to
  \End(E)$ by
  \begin{subequations}
    \begin{align}
      \theta(t) &:= \tr S(t)
      \, , \\
      \omega(t) &:= (S(t) - S(t)^\dagger) / 2
      \, , \\
      \label{eq:shear}
      \sigma(t)
      &:= (S(t) + S(t)^\dagger) / 2 - \theta(t)/n \cdot \id_E
      \,.
    \end{align}\end{subequations}
  (Recall that $n = \dim E$.)
\end{definition}

Taking the trace of the Riccati equation $S'+S^2+R = 0$ and rewriting
the quadratic term, one obtains the following result (see, e.g.,~\cite[Chap.~4]{HE}).

\begin{lemma}
  \label{lemma:raychaudhuri}
  Let $S \colon I' \to \End(E)$ be a solution of the Riccati
  equation~\eqref{eq:riccati-equation}. Then the expansion, vorticity and shear satisfy the scalar Riccati equation
  \bel{eq:raychaudhuri}
    \theta' + \frac{\theta^2}{n} +
    \tr(\omega^2) + \tr(\sigma^2) + \tr(R) = 0 \,.
  \ee
  If $S$ is self-adjoint, then~\eqref{eq:raychaudhuri} reduces to
  the form
\[
\theta' + \frac{\theta^2}{n} + \tr(\sigma^2) + \tr(R) = 0 \,.
\]
\end{lemma}

\vs
We now come to the main statement of this section.
\begin{theorem}
  \label{theorem:scalar-riccati-comparison}
  Let $R \colon \R \to \mathsf{S}(E)$ be smooth and assume that $\tr
  R \geq n \cdot \kappa$ for some $\kappa \in \R$ and
  $n=\dim E$. Furthermore, let $S \colon (0, b) \to \mathsf{S}(E)$
  be a solution of $S' + S^2 + R = 0$, and $s_\kappa \colon (0,
  b_\kappa) \to \R$ a solution of $s_\kappa' + s_\kappa^2 + \kappa =
  0$ that cannot be extended beyond $b_\kappa$. If $\lim_{t
   \searrow 0} (s_\kappa(t) - \tr S(t)/n)$ exists and is
  nonnegative, then $b \leq b_\kappa$ and
\[
\tr S(t) \leq n \cdot s_\kappa(t)
\]
  for all $t \in (0, b)$. Moreover, if equality holds for some $t_0
  \in (0, b)$, then equality also holds for all $t < t_0$. In this
  case, we also have $S(t) = s_\kappa(t) \id_E$ and $R(t) = \kappa
  \cdot \id_E$ for all $t \in (0, t_0]$.
\end{theorem}
\begin{proof}
 Set $r := \tfrac{1}{n} \left( \tr(\sigma^2) + \tr(R) \right)$. By
 the previous Lemma, $\tr S/n$ obeys the scalar Riccati equation
  \[
  \left( \frac{\tr S}{n} \right)'
  + \left( \frac{\tr S}{n} \right)^2 + r
  = 0 \,.
  \]
  Furthermore, by assumption we have
  \bel{r}
  r = \frac{\tr(\sigma^2) + \tr(R)}{n}
  \geq \frac{\tr(R)}{n}
  \geq \kappa \,.
  \ee
  Since $s_\kappa$ obeys the scalar Riccati equation $s_\kappa' +
  s_\kappa^2 + \kappa = 0$ and $\lim_{t \searrow 0} (s_\kappa(t) -
  \tr S(t)/n)$ exists and is nonnegative, we can apply
  Theorem~\ref{theorem:riccati-comparison}. Thus $b \leq b_\kappa$
  and $\tr S / n \leq s_\kappa$, as claimed.

  If equality holds for some $t_0 \in (0, b)$, then by
  Theorem~\ref{theorem:riccati-comparison}\textit{(3)} equality also
  holds for all $t < t_0$ and $r(t) = n \cdot \kappa$ for all $t \in
  (0, t_0]$. From~\eqref{r}, it follows that $\tr(\sigma(t)^2) = 0$
  and $\tr R(t) = n \cdot \kappa$ for all $t \in (0, t_0]$. By the
  definition, \eqref{eq:shear}, of $\sigma$, the fact that $\tr(\sigma(t)^2) = 0$
  implies that $\tr(S(t)^2) = n \cdot (\tr S(t))^2$. By the
  Cauchy--Schwarz inequality, this can only hold if $S(t)$ is a
  multiple of the identity. Thus $S(t) = s_\kappa(t) \id_E$ for all
  $t \in (0, t_0]$, since $\tr S(t) = n \cdot s_\kappa(t)$. Finally,
  from the Riccati equation for $S$ it follows that $R(t) = \kappa
  \cdot \id_E$ for all $t \in (0, t_0]$.
\end{proof}

\section{Lorentzian Ricci Curvature Comparison}
\label{sec:RicciComparison}

In this section, we will establish various comparison theorems for
globally hyperbolic Lorentzian manifolds with Ricci curvature bounded
from below.

\subsection{Notation and Curvature Conditions}

In the following, let $M$ be an $(n+1)$-dimensional globally
hyperbolic spacetime and $\Sigma \subset M$ a smooth, spacelike,
acausal, FCC hypersurface with signed time-separation $\tau_\Sigma
\colon M \to \R$.
\begin{definition}
    We define the \emph{future spheres and balls\/} of radius $t>0$
    around $\Sigma$ to be the sets
    \[
    S^+_\Sigma(t) := \tau_\Sigma^{-1}(t) \subset I^+(\Sigma)
    \qquad \textrm{and} \qquad
    B^+_\Sigma(t) := \hspace{-0.2cm} \bigcup_{\tau \in (0, t)}
    \hspace{-0.2cm} S^+_\Sigma(\tau) \subset I^+(\Sigma) \,.
    \]
    For convenience, we set $S^+_\Sigma(0) = \Sigma$. Furthermore, in
    order to avoid the cut locus of $\Sigma$, we set
    $\Scal^+_\Sigma(t) = S^+_\Sigma(t) \cap \Ical^+(\Sigma)$ and
    $\mathcal{B}^+_\Sigma(t) = B^+_\Sigma(t) \cap \Ical^+(\Sigma)$.
\end{definition}

By Proposition~\ref{proposition:signed-time-separation}, $\tau_\Sigma$
is a distance function on $\Ical^+(\Sigma)$, the level sets of which
are the restricted future spheres $\Scal^+_\Sigma(t)$. From the
results of Section~\ref{sec:distance-function}, on $\Ical^+(\Sigma)$,
the vector field $\n := - \gradient \tau_\Sigma$ is the
future-directed timelike unit-normal to the sets $\Scal^+_\Sigma(t)$,
and the corresponding mean curvature of the hypersurfaces
$\Scal^+_\Sigma(t)$ is given by
\bel{eq:dAlembertian-and-mean-curvature} \H_t(q) = \tr S|_q = -
\dAlembert \tau_\Sigma(q), \qquad q \in \Scal^+_\Sigma(t).
\ee

In general, future balls and spheres do not have finite volume and
area, respectively.%
\footnote{This is clear in Minkowski spacetime $\R^{n+1}$, choosing
  $\Sigma = \{0\} \times \R^n$.} Therefore,
following~\cite{EhrlichJungKim, EhrlichSanchez}, we introduce
\emph{truncated\/} spheres and balls. For $A \subseteq \Sigma$, we
set
\[
S^+_A(t)
 = \left\{
  q \in S^+_\Sigma(t)
  \mid
  \exists\, p \in A: \tau_\Sigma(q) = \tau(p, q)
 \right\} \,,
\]
i.e. $S^+_A(t) \subseteq S^+_\Sigma(t)$ consists of those points that
can be reached from $A$ by a maximizing geodesic of length $t$. We
define $B^+_A(t)$ similarly, and again we set $\Scal^+_A(t) = S^+_A(t)
\cap \Ical^+(\Sigma)$ and $\mathcal{B}^+_A(t) = B^+_A(t) \cap
\Ical^+(\Sigma)$. If $A$ is compact and $t > 0$ is sufficiently small
such that $S^+_A(t)$ does not intersect $\Cut^+(\Sigma)$, then
$S^+_A(t) \subset \Scal^+_\Sigma(t)$ is also compact, and hence has
finite area. Similarly, if $B^+_A(t)$ does not intersect the causal
cut locus of $\Sigma$, it has finite volume by the coarea
formula~\eqref{eq:coarea-formula}.

\vs The following curvature conditions will be assumed in all
comparison statements.
\begin{definition}
 \label{def:cosmological-curvature-condition}
 For constants $\kappa, \beta \in \R$, we say the pair $(M, \Sigma)$
 satisfies the \emph{cosmological comparison condition\/}
 $\CCC(\kappa, \beta)$ if the following two conditions hold.
 \begin{compactenum}
 \item $M$ has timelike Ricci curvature bounded from below by
     $\kappa$, i.e. $\Ric(v, v) \geq n \kappa$ for all $v \in TM$ with
     $\SP{v, v} = -1$.
 \item The mean curvature $\H \in \Cinfty(\Sigma)$ of $\Sigma \subset
     M$ w.r.t. $\n$ is bounded from above by $\beta$.
 \end{compactenum}
\end{definition}

\begin{remark}
    \label{rem:timelike-curvature-bounds}
{\ }\\[-4mm]
  \begin{compactenum}
  \item $M$ has timelike Ricci curvature bounded from below by $\kappa$
  if and only if for any timelike vector $v \in TM$,
  we have $\Ric(v, v) \geq - n \cdot \kappa \SP{v, v}$.
  The condition $\Ric(v, v) \geq 0$ for all timelike vectors $v \in TM$
  is also called the \emph{timelike convergence condition\/}
  or the \emph{strong energy condition\/} (cf.~\cite[pp.~95]{HE}).
  \item Recall that
  $\Ric(v, v) = \SP{v, v} \cdot \sum_{i=1}^n K(v, e_i)$,
  where $e_1, \ldots, e_n \in v^\bot$ is an orthonormal
  basis and $K(v, e_i)$ is the sectional curvature of the plane
  spanned by $v$ and $e_i$. This shows that a \emph{lower\/}
  bound on sectional curvature implies an \emph{upper\/} bound on
  timelike Ricci curvature.
\item If $(M,\g)$ is a Friedmann--Robertson--Walker spacetime,
then $\beta$ can be related to the
\emph{Hubble parameter}, i.e. the rate of acceleration of the
    universe (cf., e.g.,~\cite[pp.~433]{ONeill}). This is the reason
    for the choice of terminology in
    Def.~\ref{def:cosmological-curvature-condition}.
\end{compactenum}
\end{remark}

\subsection{Comparison Geometries}
\label{sec:comp-geom}

Here we construct certain globally hyperbolic Lorentzian manifolds
where the inequalities in the $\CCC(\kappa, \beta)$ condition become
equalities. This will lead to a suitable family of comparison spaces.

Our comparison geometries are warped products of the following form.
Let $(a, b) \subset \R$ be an interval,
$(N, \h)$ an $n$-dimensional Riemannian manifold,
and $f \in \Cinfty((a, b))$ a smooth, positive function.
We consider the Lorentzian warped product $(M, \g)$,
where $M = (a, b) \times N$ and $\g$ is given by
\[
\g = - \D t^2 + f(t)^2 \h \,.
\]
We choose the time-orientation such that $\partial_t$ is future-directed.
We take $(N, \h)$ to be complete in order that
$(M, \g)$ be globally hyperbolic (cf.~\cite[Sec.~3.6]{BEE}).
In this case, for each $t \in (a, b)$ the hypersurface
$N_t := \{t\} \times N \subset M$ is a smooth,
spacelike Cauchy hypersurface.
In particular, it is acausal and causally complete.

In order to satisfy the lower Ricci curvature bound in
Def.~\ref{def:cosmological-curvature-condition}, we will construct
comparison spaces that are Einstein, i.e. satisfy $\Ric = - n \kappa
\, \g$.%
\footnote{The negative sign appears since we want $\Ric(v, v) = n
  \kappa$ for $\SP{v, v} = -1$. (Compare Remark~\ref{rem:timelike-curvature-bounds} (1).)} A standard curvature
calculation implies that this holds if and only if $(N, \h)$ is
Einstein with $\Ric_N = (n-1) \kappa_N \h$ and the warping function
satisfies \bel{eq:einstein-warping-function} f'' = - \kappa \cdot f
\quad \textrm{and} \quad \left( f' \right)^2 + \kappa_N = f \cdot f''
\,. \ee For each $\kappa \in \R$ and given initial conditions, there
exists a unique maximal solution of the left equation. Separately,
for each $\kappa_N \in \R$ and given initial conditions, there is a
unique maximal solution of the right equation. For certain values of
$\kappa, \kappa_N \in \R$, these solutions coincide if the initial
conditions are chosen appropriately (see
Table~\ref{tab:warping-functions-for-einstein-metrics}). The two
missing cases $\kappa = 0$, $\kappa_N > 0$ and $\kappa > 0$,
$\kappa_N \geq 0$ cannot be matched.
\begin{table}
  \centering
  \begin{tabular}{llcll}
    $\kappa < 0$ & $\kappa_N > 0$ & \hspace{0.5cm} &
    $f(t) = \sqrt{\kappa_N/|\kappa|} \cosh(\sqrt{|\kappa|}t+b)$
    &\quad
    $H_t = n \sqrt{|\kappa|} \tanh(\sqrt{|\kappa|}t+b)$
    \\
    $\kappa < 0$ & $\kappa_N = 0$ & &
    $f(t) = \E^{\pm \sqrt{|\kappa|}t}$
    &\quad
    $H_t = \pm n \sqrt{|\kappa|}$ \\
    $\kappa < 0$ & $\kappa_N < 0$ & &
    $f(t) = \sqrt{|\kappa_N|/|\kappa|} \sinh(\sqrt{|\kappa|}t+b)$
    &\quad
    $H_t = n \sqrt{|\kappa|} \coth(\sqrt{|\kappa|}t+b)$
    \\ \midrule
    $\kappa = 0$ & $\kappa_N = 0$ & &
    $f(t) = \E^b = \textrm{const.}$
    &\quad
    $H_t = 0$ \\
    $\kappa = 0$ & $\kappa_N < 0$ & &
    $f(t) = \pm \sqrt{|\kappa_N|}t+b$
    &\quad
    $H_t = n/(t \pm b/\sqrt{|\kappa_N|})$
    \\ \midrule
    $\kappa > 0$ & $\kappa_N < 0$ & &
    $f(t) = \sqrt{|\kappa_N|/\kappa} \sin(\sqrt{\kappa}t+b)$
    &\quad
    $H_t = n \sqrt{\kappa} \cot(\sqrt{\kappa}t+b)$
    \\
    &&&
  \end{tabular}
  \caption{\label{tab:warping-functions-for-einstein-metrics}
   Warping functions that yield Einstein metrics. $f$ solves the
   system~\eqref{eq:einstein-warping-function}. $H_t = n
   f'(t)/f(t)$ is the (spatially constant) mean curvature of $N_t
   \subset M$.}
\end{table}
Note that, rescaling $f$ if necessary, then, without loss of generality,
we need only consider the cases $\kappa_N = 0, \pm 1$.

Regarding the second part of
Definition~\ref{def:cosmological-curvature-condition}, we note that
for each $t \in (a, b)$ the spacelike hypersurface $N_t \subset M$ is
totally umbilic and its shape operator with respect to $\partial_t$ is
given by $S_t = (f'(t)/f(t)) \id_{TN_t}$. Consequently, the
corresponding mean curvature is constant on each $N_t$ and given by
$H_t = \tr S_t = n \cdot f'(t)/f(t)$.
From~\eqref{eq:einstein-warping-function}, it follows directly that
the shape operators satisfy the Riccati equation
\[
S_t' + S_t^2 + \kappa \cdot \id_{TS_t} = 0 \,.
\]

\vs
We now concretely define our comparison geometries.
Let $\kappa, \beta \in \R$ be given.
From Table~\ref{tab:warping-functions-for-einstein-metrics},
one sees that there is a unique way of choosing $\kappa_N = 0, \pm 1$ and a solution
$f_{\kappa, \beta} \colon (a_{\kappa, \beta}, b_{\kappa, \beta}) \to \R$
of~\eqref{eq:einstein-warping-function} such that
$H_0 = n \cdot f_{\kappa, \beta}'(0)/f_{\kappa, \beta}(0) = \beta$.
Here $(a_{\kappa, \beta}, b_{\kappa, \beta}) \subseteq \R$ is
chosen to be the maximal interval containing $t=0$ on which
$f_{\kappa, \beta}$ remains strictly positive. Further, we denote by
$(N^n_{\kappa, \beta}, \h_{\kappa, \beta})$ the unique $n$-dimensional,
simply-connected space form of constant sectional curvature
$\kappa_N = 0, \pm 1$ as determined by $\kappa, \beta$.

\begin{definition}
    Given $\kappa, \beta \in \R$, we denote by $(M^{n+1}_{\kappa,
      \beta}, \g_{\kappa, \beta})$ the warped product
    \[
    M^{n+1}_{\kappa, \beta}
    := (a_{\kappa, \beta}, b_{\kappa, \beta}) \times
    N^n_{\kappa, \beta} \,,
    \qquad
    \g_{\kappa, \beta}
    := - \D t^2 + f_{\kappa, \beta}(t)^2 \h_{\kappa, \beta} \,,
    \]
    where $(N^n_{\kappa, \beta}, \h_{\kappa, \beta})$ and $f_{\kappa,
      \beta} \colon (a_{\kappa,\beta}, b_{\kappa, \beta}) \to \R$ are
    as described above. We set $\Sigma_{\kappa, \beta} := \{0\} \times
    N^n_{\kappa, \beta}$. Then $\Sigma_{\kappa, \beta} \subset
    M^{n+1}_{\kappa, \beta}$ is a smooth, spacelike, acausal, causally
    complete hypersurface of constant mean curvature $\beta$. Thus,
    for the pair $(M^{n+1}_{\kappa, \beta}, \Sigma_{\kappa, \beta})$,
    the $\CCC(\kappa, \beta)$ condition is sharp.
\end{definition}

By construction, the signed time-separation $\tau_{\kappa, \beta}
\colon M^{n+1}_{\kappa, \beta} \to \R$ of $\Sigma_{\kappa, \beta}$
agrees with the function $t:= \textrm{pr}_1 \colon M \to (a_{\kappa, \beta}, b_{\kappa, \beta})$. In particular, the future-directed
maximizing geodesics emanating from $\Sigma_{\kappa, \beta}$ are given
by the integral curves of $\partial_t = - \gradient t$. Since
integral curves do not cross, it follows that every point in
$I^+(\Sigma_{\kappa, \beta})$ is connected to $\Sigma_{\kappa, \beta}$
by a unique maximizing geodesic. Consequently, we deduce that
${\Cut}^+(\Sigma_{\kappa, \beta}) = \emptyset$. Further, from
$\tau_{\kappa, \beta} = t$, it follows that the future spheres around
$\Sigma_{\kappa, \beta}$ are the sets $S^+_{\kappa, \beta}(t) =
\Scal^+_{\kappa, \beta}(t) = \{ t \} \times N^n_{\kappa, \beta}$. As
noted previously, these hypersurfaces have constant mean curvature
$H_{\kappa, \beta}(t) = H_t = - \dAlembert{}_{\kappa, \beta}
\tau_{\kappa, \beta} |_{S^\pm_{\kappa, \beta}(t)}$ w.r.t $\partial_t$
(see Table~\ref{tab:warping-functions-for-einstein-metrics} for
$H_t$). It follows from the variation of area
formula~\eqref{eq:first-variation-of-area} that, for $B \subseteq
\Sigma_{\kappa, \beta}$, we have \bel{eq:area-in-comparison-space}
\area_{\kappa, \beta} S^+_B(t) = \frac{\area_{\kappa, \beta}
  B}{f_{\kappa, \beta}(0)^n} \cdot f_{\kappa, \beta}(t)^n \,.  \ee The
volumes of future balls are obtained by integrating this equation via
the coarea formula.

\subsection{d'Alembertian and Mean Curvature Comparison}
\label{subsec:dalembertian-mean-curvature-comparison}

We now prove the first comparison theorem. In the following
statements, quantities labelled with indices $\kappa, \beta$ belong to the
comparison geometries $(M^{n+1}_{\kappa, \beta}, \g_{\kappa, \beta})$
introduced in Sec.~\ref{sec:comp-geom}.

\begin{theorem}
  \label{theorem:dAlembertian-comparison}
  Let $\kappa, \beta \in \R$ and assume that $M$ and
  $\Sigma \subset M$ satisfy the $\CCC(\kappa, \beta)$.
  Then, for each $q \in \Ical^+(\Sigma)$,
  we have $\tau_\Sigma(q) < b_{\kappa, \beta}$ and
  \bel{eq:dAlembertian-comparison-hypersurface}
    \H_{\tau_\Sigma(q)}(q)
    = - \dAlembert \tau_\Sigma(q)
    \leq - \dAlembert{}_{\kappa, \beta}
    \tau_{\kappa, \beta}|_{S_{\kappa, \beta}(\tau_\Sigma(q))}
    = \H_{\kappa, \beta}(\tau_\Sigma(q)) \,.
  \ee
\end{theorem}
\begin{proof}
    As noted previously, we have $\H_{\tau_\Sigma(q)}(q) = \tr S|_q =
    - \dAlembert \tau_\Sigma(q)$, and \bel{star} \nabla_{\n} S + S^2 +
    R_{\n} = 0 \,, \ee where $R_{\n} = \Rm(\cdot, \n) \n$.

    Fix $q \in \Ical^+(\Sigma)$ and let $\gamma \colon [0,
    \tau_\Sigma(q)] \to M$ be the unique maximizing geodesic from
    $\Sigma$ to $q$, parametrized to unit-speed. Denote by
    $\gamma^\bot \to [0, \tau_\Sigma(q)]$ the normal bundle of
    $\gamma$, i.e. $\gamma^\bot_t = \dot{\gamma}(t)^\bot \subset
    T_{\gamma(t)}M$ for all $t \in [0, \tau_\Sigma(q)]$. Choose a
    parallel orthonormal frame $e_1, \ldots, e_n \in
    \Secinfty(\gamma^\bot)$ and let $e^1, \ldots, e^n \in
    \Secinfty((\gamma^\bot)^*)$ be the dual coframe. One can show
    that $\gamma^*S$ and $\gamma^*R_{\n}$ take values in
    $\End(\gamma^\bot)$. Therefore, we may write $\gamma^*S =
    \Scal^i_j (e_i \otimes e^j)$ and $\gamma^*R_{\n} =
    \mathcal{R}^i_j (e_i \otimes e^j)$ for smooth functions
    $\Scal^i_j, \mathcal{R}^i_j \colon [0, \tau_\Sigma(q)] \to \R$.
    Since the frames were chosen orthonormal, and both
    $R_{\mathfrak{n}}$ and $S$ are (pointwise) self-adjoint w.r.t
    $\g$, the maps $\Scal = (\Scal^i_j), \mathcal{R} =
    (\mathcal{R}^i_j) \colon [0, \tau_\Sigma(q)] \to \End(\R^n)$ are
    (pointwise) self-adjoint w.r.t. the Euclidean inner product on
    $\R^n$. Further, since $\gamma$ is an integral curve of $\n$
    (Proposition~\ref{proposition:signed-time-separation}),
    \eqref{star} implies that $\Scal' + \Scal^2 + \mathcal{R} = 0$.
    By the $\CCC(\kappa, \beta)$-assumption, we have
    \[
    \tr \mathcal{R}(t)
    = \tr\{ \Rm(\cdot, \dot{\gamma}(t))\dot{\gamma}(t) \}
    = \Ric(\dot{\gamma}(t), \dot{\gamma}(t))
    \geq n \cdot \kappa
    \]
    and
    \[
    \tr \Scal(0)
    = \tr S|_{\gamma(0)}
    = \H(\gamma(0))
    \leq \beta \,.
    \]
    On the other hand, $s_{\kappa, \beta} := \frac{1}{n} H_{\kappa,
      \beta} \colon (0, b_{\kappa, \beta}) \to \R$ satisfies the
    scalar Riccati equation $s_{\kappa, \beta}' + s_{\kappa, \beta}^2
    + \kappa = 0$ with initial conditions $s_{\kappa, \beta}(0) =
    \beta$, and cannot be extended beyond $b_{\kappa, \beta}$.
    Therefore, we can apply the scalar Riccati comparison
    theorem~\ref{theorem:scalar-riccati-comparison}, and obtain
    $\tau_\Sigma(q) < b_{\kappa, \beta}$ and $\tr S(\gamma(t)) \leq
    H_{\kappa, \beta}(t)$ for all $t \in (0, \tau_\Sigma(q)]$.
    Setting $t= \tau_\Sigma(q)$
    gives~\eqref{eq:dAlembertian-comparison-hypersurface}.
\end{proof}

\subsection{Area Comparison}
\label{subsec:some-lorentzian-area-comparison}

We now use the d'Alembertian comparison, together with the variation
of area formula~\eqref{eq:first-variation-of-area} and the coarea
formula~\eqref{eq:coarea-formula}, to obtain comparison statements
for areas and volumes of future spheres and balls.

\begin{theorem}
\label{theorem:lorentzian-area-comparison-to-hypersurface}
Let $\kappa, \beta \in \R$ and assume that $M$ and $\Sigma \subset M$
satisfy the $\CCC(\kappa, \beta)$. Then, for any $A \subseteq \Sigma$
and $B \subseteq \Sigma_{\kappa, \beta}$, the function
\[
t \mapsto \frac{\area \Scal^+_A(t)}{\area_{\kappa, \beta} S^+_B(t)},
\qquad t \in [0, b_{\kappa, \beta})
\]
is nonincreasing. Further, for $\tau \searrow 0$, this ratio
converges to $\area A / \area_{\kappa, \beta} B$, so we also have
\[
\area \Scal^+_A(t)
\leq \frac{\area A}{\area_{\kappa, \beta} B} \cdot \area_{\kappa, \beta} S^+_{B}(t)\, ,
\]
for all $t \in [0, b_{\kappa, \beta})$.
\end{theorem}

\begin{remark}
\label{remark:measurability}
For general $A \subseteq \Sigma$, the sets $\Scal^+_A(t) \subseteq
\Scal^+(\Sigma, t)$ may not be measurable. In this case, $\area
\Scal^+_A(t)$ should be understood as the inner measure given by
$\sup_K \left( \area K \right)$, where the supremum is taken over all
compact sets $K \subset \Scal^+_A(t)$. (This will be clear from the
proof below.) Since Riemannian measures are Radon measures, this
gives the correct result in the measurable case. The same remark
applies in the following statements.
\end{remark}

\begin{proof}
 Let $0 < t_1 < t_2 < b_{\kappa, \beta}$. Choose a sequence of
 compact sets $K_i \subset \Scal^+_A(t_2)$ with $\area K_i
 \nearrow \area \Scal^+_A(t_2)$. Each point in
 $\Scal^+_A(t_2)$ can be reached from $\Sigma$ by a unique
 maximizing, future-directed unit-speed geodesic. Since these
 geodesics are integral curves of $\n = -\gradient \tau_\Sigma$
 (Proposition~\ref{proposition:signed-time-separation}), we have
 \[
 K_i(t)
 := \Phi_{t-t_2}(K_i)
 \subset \Scal^+_A(t) \,,
 \]
 where $\Phi$ is the flow of $\n$. Further, for each $i \in
 \mathbb{N}$ and each $t \in [0, t_2]$, $K_i(t) \subset
 \Scal^+_A(t)$ is compact and $\Phi$ is defined on $(-t,
 t_2-t) \times K_i(t)$. Therefore, we may use the variation of area
 formula~\eqref{eq:first-variation-of-area},
 and~\eqref{eq:dAlembertian-comparison-hypersurface}, giving
 \[
  \frac{\D}{\D t} \log \left( \area K_i(t) \right)
  = \frac{1}{\area K_i(t)} \int_{K_i(t)} \H_t(q) d \mu_t(q)
  \leq H_{\kappa, \beta}(t)
  \stackrel{~\eqref{eq:area-in-comparison-space}}{=}
  \frac{\D}{\D t} \log \area_{\kappa, \beta} S^+_{B}(t)\,.
 \]
 This shows that the function $t \mapsto \area K_i(t) /
 \area_{\kappa, \beta} S^+_{B}(t)$ is nonincreasing on $[0, t_2]$. Hence
 \[
   \frac{\area K_i(t_2)}{\area_{\kappa, \beta} S^+_{B}(t_2)}
  \leq \frac{\area K_i(t_1)}{\area_{\kappa, \beta} S^+_{B}(t_1)}
  \leq
  \frac{\area \Scal^+_A(t_1)}
  {\area_{\kappa, \beta} S^+_{B}(t_1)}
  \,,
  \]
  where the final inequality is simply due to the inclusion $K_i(t_1)
  \subseteq \Scal^+_A(t_1)$. For $i \to \infty$, this yields
  \[
  \frac{\area \Scal^+_A(t_2)}
  {\area_{\kappa, \beta} S^+_{B}(t_2)}
  \leq
  \frac{\area \Scal^+_A(t_1)}
  {\area_{\kappa, \beta} S^+_{B}(t_1)} \,.
  \]
  This shows monotonicity. The second assertion is clear.
\end{proof}

A special case of this result is the following.
\begin{corollary}
    Let $\kappa, \beta \in \R$ and assume that $M$ and $\Sigma
    \subset M$ satisfy the $\CCC(\kappa, \beta)$. Let $A \subseteq
    \Sigma$ and $B \subseteq \Sigma_{\kappa, \beta}$ with the
    property that $\area_{\kappa, \beta} B = \area A$. Then, the
    function
    \[
    t \mapsto \frac{\area \Scal^+_A(t)}{\area_{\kappa, \beta} S^+_B(t)}, \qquad t \in [0, b_{\kappa, \beta})
    \]
    is nonincreasing, and converges to $1$ as $t \searrow 0$. Therefore,
    \[
    \area \Scal^+_A(t) \leq \area_{\kappa, \beta} S^+_{B}(t), \qquad
    t \in [0, b_{\kappa, \beta}).
    \]
\end{corollary}

\subsection{Volume Comparison}
\label{subsec:some-lorentzian-volume-comparison}

Using the coarea formula~\eqref{eq:coarea-formula} and the following
Lemma, the area comparison theorem immediately yields a volume
comparison result.
\begin{lemma}
  \label{lemma:gromov-lemma}
  Let $f, g \colon [a, b) \to [0, \infty)$ be locally integrable,
  nonzero on $(a, b)$, and assume that $f/g$ is non-increasing on
  $(a, b)$. Then the functions $F, G \colon (a, b) \to (0, \infty)$, defined by
  \[
  F(x) = \int_a^x f(y) \D y
  \qquad \textrm{and} \qquad
  G(x) = \int_a^x g(y) \D y \,,
  \]
  are continuous, and $F/G$ is also non-increasing on $(a, b)$.
\end{lemma}
\begin{proof}
 Since $f$ and $g$ are locally integrable, $F$ and $G$ are
 well-defined and continuous. The rest of the proof may be found
 in~\cite[pp.~42]{CGT}.
\end{proof}

\begin{theorem}
  \label{theorem:lorentzian-volume-comparison-to-hypersurface}
  Let $\kappa, \beta \in \R$ and assume that $M$ and $\Sigma \subset
  M$ satisfy the $\CCC(\kappa, \beta)$. Then, for any $A \subseteq
  \Sigma$ and $B \subseteq \Sigma_{\kappa, \beta}$, the function
  \[
  t \mapsto \frac{\vol B^+_A(t)} {\vol_{\kappa, \beta} B^+_{B}(t)} ,
  \qquad t \in [0, b_{\kappa, \beta})
  \]
  is nonincreasing. Further, for $t \searrow 0$, this ratio converges
  to $\area A / \area_{\kappa, \beta} B$, so we also have
  \[
  \vol B^+_A(t) \leq \frac{\area A}{\area_{\kappa, \beta} B} \cdot \vol_{\kappa, \beta} B^+_{B}(t)
  \]
  for all $t \in [0, b_{\kappa, \beta})$.
\end{theorem}
\begin{proof}
  By the coarea formula~\eqref{eq:coarea-formula}, for any
  $t \in [0, b_{\kappa, \beta})$, we have
  \bel{star2}
  \vol B^+_A(t)
  = \int_0^t \area \Scal^+_A(\tau) \D \tau \,.
  \ee
  Let $0 < t_1 < t_2 < b_{\kappa, \beta}$ be given. We distinguish
  two cases. First, assume that $\vol B^+_A(t_2) = \infty$. Then, by~\eqref{star2}, there exists $\tau_0 \in (0, t)$ with $\area
  \Scal^+_A(\tau_0) = \infty$. By area comparison, we
  therefore must have $\area \Scal^+_A(\tau) = \infty$ also
  for all $\tau < \tau_0$. By~\eqref{star2} again, it follows that also
  $\vol B^+_A(t_1) = \infty$, hence the assertion is trivially
  satisfied.

  Now assume that $\vol B^+_A(t_2) < \infty$. Then by~\eqref{star2}, the
  function $\tau \mapsto \area \Scal^+_A(\tau)$ is locally
  integrable on $[0, t_2]$. Since also $\tau \mapsto
  \area_{\kappa, \beta} S^+_{B}(t)$ is locally integrable, we may
  apply Lemma~\ref{lemma:gromov-lemma} together with
  the area comparison
  Theorem~\ref{theorem:lorentzian-area-comparison-to-hypersurface}. This
  yields the monotonicity assertion.

Finally, as $t \searrow 0$, we use~\eqref{star2} and l'H\^{o}pital's rule to obtain
  \begin{align*}
    \lim_{t \searrow 0}
    \frac{\vol B^+_A(t)}
    {\vol_{\kappa, \beta} B^+_{B}(t)}
    = \lim_{t \searrow 0}
    \frac{\area \Scal^+_A(t)}
    {\area_{\kappa, \beta} S^+_{B}(t)}
    = \frac{\area A}{\area_{\kappa, \beta} B} \,.
  \end{align*}
\end{proof}

Again, we state a special case of this result.

\begin{corollary}
Let $\kappa, \beta \in \R$ and $(M, \Sigma)$ satisfy the $\CCC(\kappa, \beta)$.
Let $A \subseteq \Sigma$ and $B \subseteq \Sigma_{\kappa, \beta}$ be such that $\area_{\kappa, \beta} B = \area A$.
Then, the function
\[
t \mapsto \frac{\vol B^+_A(t)} {\vol_{\kappa, \beta} B^+_{B}(t)},
\qquad t \in [0, b_{\kappa, \beta})
\]
is nonincreasing, and converges to $1$ as $t \searrow 0$. Hence,
\[
\vol B^+_A(t) \leq \vol_{\kappa, \beta} B^+_{B}(t), \qquad t \in [0, b_{\kappa, \beta}).
\]
\end{corollary}

\section{Application to Singularity Theorems}
\label{sec:singularity-theorems}

We now use the comparison results of the previous section to prove
the following singularity theorem due to Hawking~\cite[pp.~272]{HE}.

\begin{theorem}
Let $M$ be globally hyperbolic and $\Sigma \subset M$ a smooth, spacelike, acausal, FCC hypersurface.%
\footnote{For example, $\Sigma \subset M$ could be a smooth, spacelike Cauchy hypersurface.}
Assume that $M$ and $\Sigma$ satisfy the $\CCC(\kappa, \beta)$ with $\kappa = 0$ and $\beta < 0$.
Then no future-directed curve starting in $\Sigma$
can have arc-length greater than $1/|\beta|$.
In particular, $M$ is timelike geodesically incomplete.
\end{theorem}

\begin{proof}[Proof via d'Alembertian Comparison]
This proof is based on the proof of Myers's theorem in Riemannian geometry given in~\cite{zhu:1997a}.

Let $\gamma \colon [0, b] \to M$ be a maximizing, timelike,
future-directed, unit-speed geodesic emanating perpendicular from
$\Sigma$. Then we have $\gamma(t) \in \Ical^+(\Sigma)$ for all
$t \in (0, b)$. For $\kappa = 0$ and $\beta < 0$, the d'Alembertian
comparison Theorem~\ref{theorem:dAlembertian-comparison} yields
\[
- (\dAlembert \tau_\Sigma)(\gamma(t))
\leq \H_{0, \beta}(t)
= \frac{1}{t+1/\beta}
= \frac{1}{t-1/|\beta|} \,,
\]
for all $t \in (0, b)$.
Since the right hand side diverges to $-\infty$ for $t \nearrow 1/|\beta|$,
but the left hand side is finite for all $t \in (0, b)$,
this implies that $b \leq 1/|\beta|$.
Thus, since every point in $I^+(\Sigma)$ can be connected to $\Sigma$ by a
maximizing geodesic, we have $\tau_\Sigma(q) \leq 1/|\beta|$ for all
$q \in I^+(\Sigma)$. From the definition of the time-separation, this gives
the required upper bound on arc-length of future-directed curves starting in
$\Sigma$. Timelike geodesic incompleteness follows immediately.
\end{proof}

\begin{proof}[Proof via Area Comparison]

We will show that $S^+_\Sigma(1/|\beta|) \subset
\Cut^+(\Sigma)$, which implies that $S^+_\Sigma(t) = \emptyset$
for all $t > |\beta|$. This again yields $\tau_\Sigma(q) \leq
1/|\beta|$ for all $q \in I^+(\Sigma)$, so we can proceed as in the
previous proof.

For the sake of contradiction, we assume that there exists
$q \in S^+_\Sigma(1/|\beta|) \setminus \Cut^+(\Sigma) = \Scal^+_\Sigma(1/|\beta|)$.
Since the cut locus is closed, there exists a neighborhood
$K \subset \Scal^+_\Sigma(1/|\beta|)$ of $q$ with the property that $\area K > 0$.
Set $A := \Phi_{-1/|\beta|}(K) \subseteq \Sigma$, where $\Phi$ is
the flow of $\n = - \gradient \tau_\Sigma$, and choose any subset
$B \subseteq \Sigma_{0, \beta}$.
Then, by the area comparison theorem~\ref{theorem:lorentzian-area-comparison-to-hypersurface},
we obtain
\[
\area S^+_A(t)
\leq \frac{\area A}{\area_{0, \beta} B}
\cdot \area_{0, \beta} S^+_{B}(t)
\sim \big( 1 - |\beta|t \big)
\]
for all $t \in (0, 1/|\beta|)$. It follows that $\area S^+_A(t_0) = 0$
for some $t_0 \le 1/|\beta|$, and therefore $\area K \leq \area
S^+_A(1/|\beta|) = 0$. This contradicts the choice of $K$.
\end{proof}

\section{Final remarks}

Our results should be compared with corresponding results in
Riemannian geometry. In particular, our proof of the
singularity theorem is largely analogous to the proof of
Myers's Theorem, which states that a complete Riemannian
manifold (of dimension $n$) that satisfies the lower Ricci
curvature bound $\Ric \geq (n-1) \kappa \g$ for some constant
$\kappa > 0$, is necessarily compact, with diameter less than or equal to $\pi/\sqrt{\kappa}$.%
\footnote{In particular, adapting our techniques to develop
  comparison results for a point instead of a hypersurface, and
  assuming the stronger curvature bound $\Ric \geq \kappa > 0$, one
  can obtain results on existence of conjugate points that are more
  closely related to Myers's theorem.} In a more speculative
direction, the area-theoretic approach to the singularity theorems
given above may be applicable in more general situations where the
Lorentzian metric is of low regularity. For metrics that are not
$C^2$, one may \emph{define\/} a Ricci curvature bound in terms of
monotonicity properties of area functionals along geodesics (see,
e.g.,~\cite{Ohta}). Whether one can then develop suitable
synthetic-geometrical techniques in Lorentzian geometry to prove
singularity theorems for low-regularity metrics is currently under
investigation~\cite{SynthSing}.

\appendix
\section{Curvature bounded above}
\label{sec:bddabove}

Let $(M, \g)$ be a Lorentzian manifold and $\Sigma$ a smooth
spacelike hypersurface in $M$. We finally investigate some conditions
under which we may prove a \emph{lower\/} bound on the time
separation between $\Sigma$ and focal points along normal geodesics.

\begin{proposition}
\label{Rauch:submanifold}
Let $\gamma$ be a normal geodesic to $\Sigma$ parametrized by arc-length.
Let $\kappa, \beta \in \R$ be constants such that the curvature operator
$R_{\gammad} := \Rm(\cdot, {\gammad}) {\gammad}$ satisfies $R_{\gammad} \leq \kappa \id$,
and the shape operator of $\Sigma$ satisfies $S_\Sigma \geq \beta \id$.
Define the positive constant $t_0 = t_0(\kappa, \beta)$
to be the first positive value of $t$ for which the following equations hold:
\begin{subequations}
\begin{align}
\cot \left( \sqrt{\kappa} t \right) &= - \frac{\beta}{\sqrt{\kappa}} &\kappa&>0,
\\
t &= - \frac{1}{\beta} &\kappa&=0,
\\
\coth \left( \sqrt{|\kappa|} t \right) &= - \frac{\beta}{\sqrt{|\kappa|}} &\kappa&<0.
\end{align}\label{zeroes}\end{subequations}
(If there are no solutions for positive $t$, set $t_0 = +\infty$.)
Then no point $\gamma(t)$ along the geodesic $\gamma$ is a
focal point of $\Sigma$ for $0 < t < t_0$.
\end{proposition}

\begin{proof}
    We proceed in a similar way to the proof of Theorem~\ref{theorem:dAlembertian-comparison}. 
    First, let
    $s_{\kappa,\beta} \colon [0,b_{\kappa,\beta}) \to \R$ denote the
    maximal solution of $s_{\kappa,\beta}' + s_{\kappa,\beta}^2 +
    \kappa = 0$ with $s_{\kappa,\beta}(0) = \beta$. We then have
    $s_{\kappa,\beta} = 1/n \cdot H = f'/f$, where the functions $H$
    and $f$ may be found in Table~\ref{tab:warping-functions-for-einstein-metrics} 
    (with the constants chosen appropriately). One may check that $t_0$ as
    defined by~\eqref{zeroes} corresponds precisely to the first
    positive zero of $f$, and hence coincides with $b_{\gamma,\beta}$.

    Next, choose an orthonormal frame $e_1, \ldots, e_n \in
    \Secinfty(\gamma^\bot)$ for the normal bundle of $\gamma$, and
    let $e^1, \ldots, e^n \in \Secinfty((\gamma^\bot)^*)$ be the dual
    coframe. As in the proof of
    Theorem~\ref{theorem:dAlembertian-comparison}, we note that
    $R_{\gammad}$ and the shape operator of the future spheres, $S(X)
    = - \nabla_X \gradient \tau_\Sigma$, only take values
    perpendicular to $\gamma$. Therefore, we have $R_\gamma =
    \mathcal{R}^i_j e_i \otimes e^j$ and $\gamma^*S = \Scal^i_j e_i
    \otimes e^j$, where the smooth maps $\mathcal{R} =
    (\mathcal{R}^i_j), \Scal = (\Scal^i_j) \colon [0, T]
    \to \End(\R^n)$ are self-adjoint w.r.t the Euclidean inner
    product and satisfy the Riccati equation $\Scal' + \Scal^2 +
    \mathcal{R} = 0$.

    Our assumptions are equivalent to $\mathcal{R} \leq \kappa \id$
    and $\Scal(0) \geq \beta \id$. Therefore, a direct application of
    Theorem~\ref{theorem:riccati-comparison} implies that
    \bel{shapelowerbound} \Scal(t) \geq s_{\kappa,\beta}(t) \cdot \id
    \,. \ee Since focal points of $\Sigma$ along $\gamma$ correspond
    precisely to points at which $\Scal$ becoming singular in the
    sense that $\tr \Scal(t) \searrow - \infty$, it follows
    from~\eqref{shapelowerbound} and the observation at the beginning
    of the proof that this situation cannot occur before $t_0$ as
    defined by~\eqref{zeroes}. This finishes the proof.
\end{proof}

\begin{remarks}
{\ }\\[-4mm]
\begin{compactenum}
\item The conditions of Proposition~\ref{Rauch:submanifold} may
    alternatively be stated as saying that
    \bel{LorentzianCurvatureBound} \langle \Rm(X, \dot{\gamma})
    \dot{\gamma}, X \rangle \le \kappa \left( \langle \gammad,
        \gammad \rangle \langle X, X \rangle - \langle \gammad, X
        \rangle^2 \right) \ee along $\gamma$, for all vector fields
    $X$ defined along $\gamma$, and that the eigenvalues of the
    second fundamental form of $\Sigma$ at $p = \gamma(0)$ are
    bounded below by $\beta$. In this form,
    Proposition~\ref{Rauch:submanifold} is essentially an adaption to
    Lorentzian geometry of the Rauch comparison theorem for
    submanifolds of Riemannian manifolds~\cite{Warner:Rauch}.
\item The estimates in Proposition~\ref{Rauch:submanifold} are
    \emph{sharp}, with equality being achieved for hypersurfaces with
    all eigenvalues of the shape operator equal to $\beta$ in the
    two-dimensional, model Lorentzian manifold of constant curvature
    $\kappa$.
\item Note that the constant $t_0$ is independent of the dimension of
    the manifold $M$.
\end{compactenum}
\end{remarks}

Applying Proposition~\ref{Rauch:submanifold} along all
geodesics normal to $\Sigma$, we have the following result.

\begin{theorem}
\label{thm:noconjpts} Let $\Sigma \subset M$ be a spacelike
hypersurface. Let $\kappa, \beta \in \R$ be given constants.
Assume that, for any future-directed geodesic normal to
$\Sigma$, $\gamma \colon [0, T] \to M$, normalised such that
$\SP{\gammad, \gammad} = -1$, the curvature operator
$R_{\gammad}$ satisfies $R_{\gammad} \leq \kappa \id$. Assume
further that the shape operator of $\Sigma$ satisfies $S_\Sigma
\geq \beta \id$. Then no point along $\gamma$ is a focal point
of $\Sigma$ if $T < t_0$, where $t_0$ is as in~\eqref{zeroes}.
\end{theorem}

We now note that the proof of Proposition~\ref{zeroes} yields the following result.

\begin{proposition}
    Let $\kappa, \beta \in \R$ and that $(M, \g)$ satisfy the
    conditions of Theorem~\ref{thm:noconjpts}. Then, for $t > 0$
    sufficiently small such that $S^+_{\Sigma}(t) \cap \Cut^+(\Sigma)
    = \emptyset$, the mean curvature $H_t$ of $S^+_{\Sigma}(t)$
    satisfies \bel{Hbound} \H_t \ge \H_{\kappa,
      \beta}(\tau_\Sigma(q)) \,, \ee where $\H_{\kappa, \beta}$ are
    the functions given in
    Table~\ref{tab:warping-functions-for-einstein-metrics}.
\end{proposition}
\begin{proof}
Take the trace of~\eqref{shapelowerbound}.
\end{proof}

Following through the proof of Theorem~\ref{theorem:lorentzian-area-comparison-to-hypersurface}, we have the following result.

\begin{theorem}
\label{thm:areaCBA}
Let $\kappa, \beta \in \R$ and assume that $(M, \g)$ and $\Sigma
\subset M$ satisfy the conditions of Theorem~\ref{thm:noconjpts}.
Then, for any $A \subseteq \Sigma$ and $B \subseteq \Sigma_{\kappa,
  \beta}$, and $t > 0$ sufficiently small that $\Scal^+_A(t) \cap
\Cut^+(\Sigma) = \emptyset$, the map
\[
t \to \frac{\area \Scal^+_A(t)}{\area_{\kappa, \beta} S^+_B(t)}
\]
is non-decreasing. Further, for $\tau \searrow 0$, this ratio
converges to $\area A / \area_{\kappa, \beta} B$, so we have
\[
\area \Scal^+_A(t) \ge \frac{\area A}{\area_{\kappa, \beta} B} \cdot \area_{\kappa, \beta} S^+_{B}(t)\, ,
\]
for all $t \in [0, b_{\kappa, \beta})$.
\end{theorem}

Finally, Lemma~\ref{lemma:gromov-lemma} has no analogue for non-decreasing functions. Therefore, as is standard, there is no relative volume monotonicity theorem in the case of curvature bounded above. Theorem~\ref{thm:areaCBA} and the coarea formula, however, yield the following volume comparison result.

\begin{theorem}
Let $\kappa, \beta \in \R$ and assume that $(M, \g)$ and $\Sigma \subset M$ satisfy the conditions of Theorem~\ref{thm:noconjpts}. Then, for any $A \subseteq \Sigma$ and $B \subseteq \Sigma_{\kappa, \beta}$, and $t > 0$ sufficiently small that $\Scal^+_A(t) \cap \Cut^+(\Sigma) = \emptyset$, we have
\[
\vol B^+_A(t) \ge \frac{\area A}{\area_{\kappa, \beta} B} \cdot \vol_{\kappa, \beta} B^+_{B}(t)
\]
for all $t \in [0, b_{\kappa, \beta})$.
\end{theorem}

\end{document}